\documentclass[aps,pra,twocolumn,twoside,a4]{revtex4}

\usepackage{graphicx,epsfig}
\usepackage{color}
\usepackage[usenames,dvipsnames]{xcolor}
\usepackage{amsmath,bbm,amssymb, amsthm}

\newcommand{\tr}{{\rm tr}}

\def\be{\begin{eqnarray}}
\def\ee{\end{eqnarray}}

\newcommand{\ket}[1]{| #1 \rangle}
\newcommand{\bra}[1]{\langle #1 |}
\newcommand{\pure}[1]{| #1 \rangle\langle #1|}
\newcommand{\avg}[1]{\langle #1 \rangle} 

\newcommand{\upix}[1]{^{( #1)}} 

\newcommand{\abs}{{\rm abs}}
\newcommand{\opt}{{\rm opt}}
\newcommand{\Dmin}{D_{\rm min}}
\newcommand{\Dmax}{D_{\rm max}}

\newcommand{\eps}{\varepsilon}

\newcommand{\Dmineps}{D_{\rm min}^\varepsilon}
\newcommand{\Dmaxeps}{D_{\rm max}^\varepsilon}

\newcommand{\ephi}{|\langle e_0|0\rangle|^2}

\newtheorem{defi}{Definition}\def\DE{\begin{defi}}\def\ED{\end{defi}}
\newtheorem{lemma}{Lemma}\def\LE{\begin{lemma}}\def\EL{\end{lemma}}
\newtheorem{theo}{Theorem}\def\LE{\begin{theo}}\def\EL{\end{theo}}

\begin{document}

\title{Coherence and measurement in quantum thermodynamics}

\author{Philipp Kammerlander}
\email{kammerlander@phys.ethz.ch}
\affiliation{Institute for Theoretical Physics, ETH Zurich, Wolfgang-Pauli-Strasse 27, 8093 Zurich, Switzerland.}

\author{Janet Anders}
\email{janet@qipc.org}
\affiliation{Department of Physics and Astronomy, University of Exeter, Stocker Road, EX4 4QL, United Kingdom.}


\begin{abstract} 

Thermodynamics is a highly successful macroscopic theory widely used across the natural sciences and for the construction of everyday devices, from car engines and fridges to power plants and solar cells. With thermodynamics predating quantum theory, research now aims to uncover the thermodynamic laws that govern finite size systems which may in addition host quantum effects.
Here we identify information processing tasks, the so-called ``projections'', that can only be formulated within the framework of quantum mechanics. We show that the physical realisation of such projections can come with a non-trivial thermodynamic work \emph{only} for quantum states with coherences. This contrasts with information erasure, first investigated by Landauer, for which a thermodynamic work cost applies for classical and quantum erasure alike. Implications are far-reaching, adding a thermodynamic dimension to measurements performed in quantum thermodynamics experiments, and providing key input for the construction of a future quantum thermodynamic framework. Repercussions are discussed for quantum work fluctuation relations and thermodynamic single-shot approaches.

\end{abstract}



\maketitle


\section{Introduction} \label{sec:intro}
When  Landauer argued in 1961 that any physical realisation of  erasure of information has a fundamental thermodynamic work cost he irrevocably linked thermodynamics and information theory \cite{Landauer61}. A practical consequence of this insight is that all computers must dissipate a minimal amount of heat in each irreversible computing step, a threshold that is becoming a concern with future computer chips entering atomic scales. The treatment of general \emph{quantum} information processing tasks within the wider framework of quantum thermodynamics has only recently begun.
Theoretical breakthroughs include the characterisation of the efficiency of quantum thermal engines \cite{Scully, Kosloff14, Lutz14} and the extension of widely used classical non-equilibrium fluctuation theorems to the quantum regime \cite{Mukamel03, TLH07}. A new thermodynamic resource theory \cite{Janzing00} has led to the discovery of a \emph{set} of second laws that replaces the standard macroscopic second law for finite size systems \cite{Brandao13b, Lostaglio14}. These results have substantially advanced our understanding of nanoscale thermodynamics, however putting a finger on what is genuinely ``quantum'' in quantum thermodynamics has remained a challenge.
Quantum mechanics differs from classical mechanics in at least three central aspects: the special nature of measurement, the possibility of a quantum system to be in a superposition and the existence of quantum correlations. The thermodynamic energy needed to perform a (selective) measurement has been investigated \cite{Jacobs12} and the total work for a closed thermodynamic measurement cycle explored \cite{Erez10}. The catalytic role of quantum superposition states when used in thermal operations has been uncovered \cite{Aberg14} and it has been shown that work can be drawn from quantum correlations \cite{Zurek03,delRio} in a thermodynamic setting, see Fig.~\ref{fig:overviewpic}. In particular, del Rio {\it et al.} \cite{delRio} showed that contrary to Landauer's principle, it is possible to \emph{extract} work while performing erasure of a system's state when the system is correlated to a memory. This can occur if and only if the initial correlations imply a negative conditional entropy, a uniquely quantum feature. The thermodynamic process does however now require operation on degrees of freedom external to the system, i.e. the memory's. 

\begin{figure}[t]
	\includegraphics[width=0.45\textwidth]{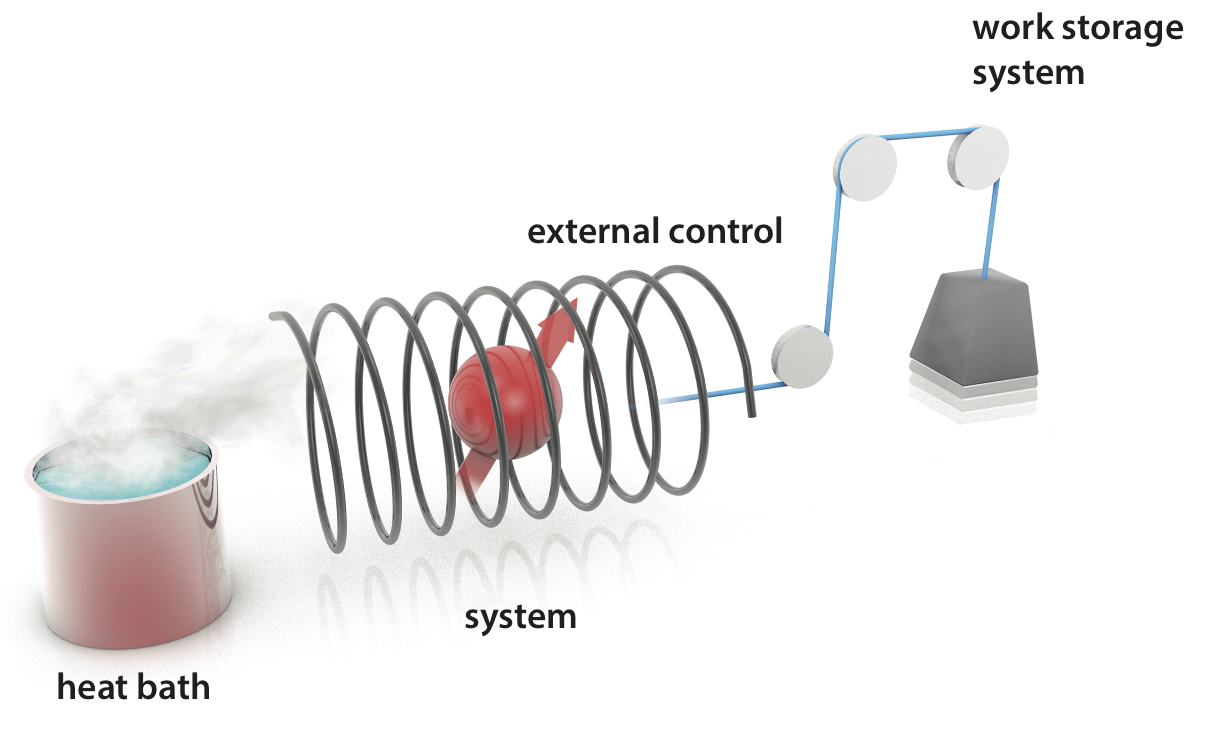}
        \caption{\label{fig:overviewpic} 
      	{\bfseries Thermodynamic setting.} 
        A system, depicted as a spin, interacts with a heat bath at temperature $T$, with which it exchanges \emph{heat}, and with controlled energy sources, illustrated as coil and weight, with which it exchanges \emph{work}. Work drawn from the system can be collected in a work storage system (weight) for future use.
        }
\end{figure}


Our motivation is here to shed light on the implications of performing a measurement on a quantum state that has coherences. We will consider this task in the thermodynamic setting of Landauer's erasure, involving a heat bath at fixed temperature $T$ and operation on $N \to \infty$ uncorrelated and identically prepared copies of the system (i.i.d. limit). This is of interest in the context of the quantum Jarzynski equality, for example, and will also be central for experiments testing quantum thermodynamic predictions in the future. To tackle this question we define the information-theoretic ``projection'' $\rho \to \eta^{\cal P} := \sum_k \Pi^{{\cal P}}_k \, \rho \, \Pi^{{\cal P}}_k$ for a given initial quantum state $\rho$ and a complete set of mutually orthogonal projectors $ \{\Pi^{{\cal P}}_k\}_k$.  Such state transformation can be seen as analogous to the state transfer of erasure, $\rho \to \ket{0}$, to a blank state $\ket{0}$. Physically, this projection can be interpreted as the result of an unread, or unselective \cite{Kurizki08}, measurement of an observable ${\cal P}$ that has eigenvector projectors $\{ \Pi^{{\cal P}}_k \}_k$.  In an unselective measurement the individual measurement outcomes are not recorded and only the statistics of outcomes is known. In the literature the implementation of unselective measurements is often not specified, although it is typically thought of as measuring individual outcomes, e.g. with a Stern-Gerlach experiment, see Fig.~\ref{fig:Blochpicture}a, followed by mixing.
The crux is that the information-theoretic projection $\rho \to \eta^{\cal P}$ can be implemented in many physical ways. The associated thermodynamic heat and work will differ depending on \emph{how} the projection was done and we will refer to the various realisations as ``thermodynamic projection processes''. One possibility is decohering \cite{decohering} the state in the so-called pointer basis, $\{\Pi^{\rm pointer}_k\}_k$, a thermodynamic process where an environment removes coherences in an uncontrolled manner resulting in no associated work. In general it is possible to implement the state transfer in a finely controlled fashion achieving optimal thermodynamic heat and work values. 

\begin{figure*}[t]
\hspace{-0.3cm}
\includegraphics[width=.99\textwidth]{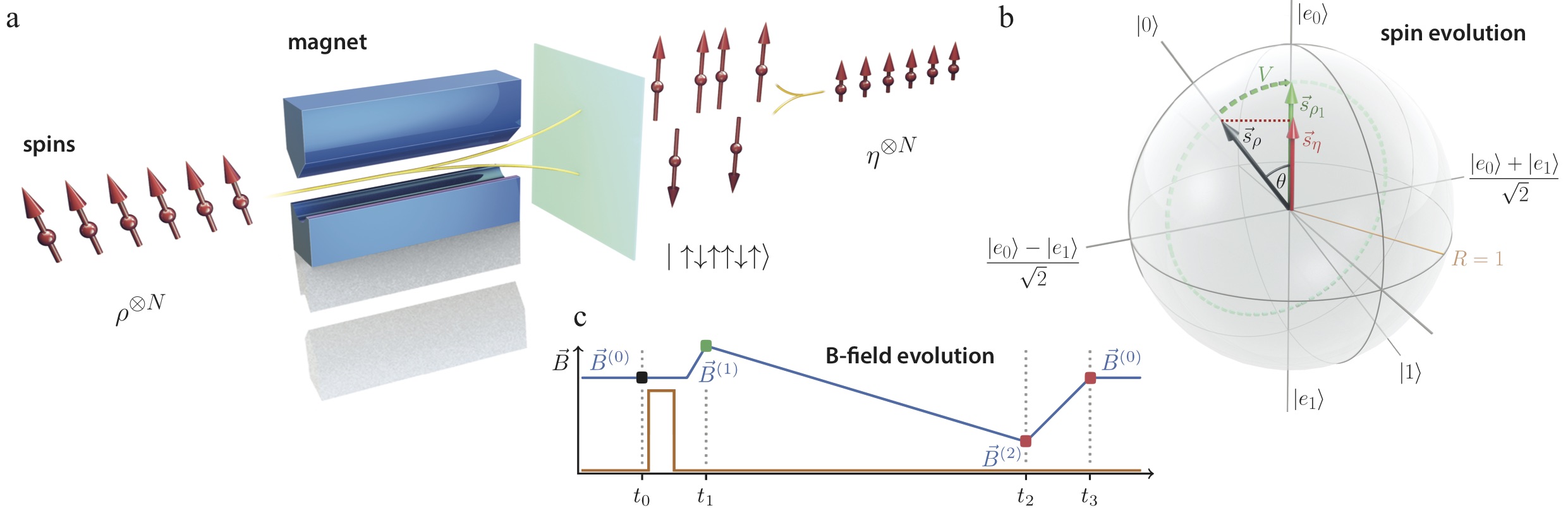}
       \caption{\label{fig:Blochpicture}
   {\bfseries Two physical realisations of a projection process.} 
           {\bfseries a,} $N$ identically prepared spin 1/2 particles in state $\rho^{\otimes N}$ pass a Stern-Gerlach magnet and a screen after which they emerge in either the spin-up or the spin-down beam. Recombining the two beams mixes the spins to the final state $\eta^{\otimes N}$ for $N \to \infty$.
Illustration of the spin example, showing the state evolution in {\bfseries b} and the B-field evolution in {\bfseries c}. 
The poles in the Blochsphere ({\bfseries b}) are the energy eigenstates $\ket{e_0}$ and $\ket{e_1}$ that are aligned and anti-aligned with an externally applied B-field (indicated in blue in {\bfseries c}), which initially is $\vec{B}\upix 0$ (black point in {\bfseries c}). 
In the first step the Blochvector $\vec{s}_\rho$ (black arrow in {\bfseries b}) of Emmy's initial state $\rho$ is rotated on the green-dashed circle to $\vec{s}_{\rho_1}$ (green arrow in {\bfseries b}).
The unitary rotation $V$ required for this step can be realised by applying a microwave pulse creating an additional B-field (indicated in orange in {\bfseries c}) in the direction orthogonal to the plane of the green circle. At the end of the first step the pulse is turned off and the external B-field is adjusted to $\vec{B}\upix 1$ (green point in {\bfseries c}).
The second step shortens $\vec{s}_{\rho_1}$ to $\vec{s}_{\eta}$ (red arrow in {\bfseries b}), the Blochvector of $\eta$ (superscripts $H$ have been omitted). The external B-field (blue in {\bfseries c}) decreases slowly to $\vec{B}\upix 2$ (red point at $t_2$ in {\bfseries c}). In the last step the B-field quickly returns to its initial value, $\vec{B}\upix 0$ (red point at $t_3$ in {\bfseries c}), while the state remains $\eta$. The angle between the Blochvectors of $\rho$ and $\eta$ is indicated by $\theta$.
} 
\end{figure*}

\section{Main result} \label{sec:result}

Of particular importance in thermodynamics is the projection $\rho \to \eta^H$ of the system's initial state $\rho$ onto the set of energy eigenstates $\{ \Pi^{H}_k \}_k$ of the system's Hamiltonian $H = \sum_k E_k \, \Pi^{H}_k$ with $E_k$ the energy eigenvalues. Here the state's off-diagonals with respect to the energy eigenbasis are removed - a state transformation that is frequently employed in quantum thermodynamic derivations and referred to as ``dephasing'' or ``measuring the energy''.
Our key observation is that there exists a thermodynamic projection process realising this transformation and allowing to draw from the quantum system a non-trivial  \emph{optimal average work} of
\begin{align} \label{eq:Wmeasure}
    \avg{W_\opt}  = k_B \, T \, \left( S (\eta^H) - S(\rho) \right).
\end{align}
Here $T$ is the temperature of the heat bath with which the system is allowed to interact, see illustration Fig.~\ref{fig:overviewpic}, $k_B$ is the Boltzmann constant and $S$ is the von Neumann entropy. Crucially, this work is strictly positive for quantum states with coherences. Extending the key observation to general projections $\rho \to \eta^{\cal P}$ one finds that  optimal thermodynamic projection processes can be implemented that allow to draw an average work of 
\begin{align}\label{eq:Wmeasuregeneral}
	\avg{W_\opt} = k_B \, T\,  \left( S (\eta^{\cal P} ) - S(\rho) \right) - \tr \left[H \, (\eta^{\cal P} -\rho)\right],
\end{align}
where an additional internal energy change term appears. 

The optimal work values stated in Eqs.~(\ref{eq:Wmeasure}) and (\ref{eq:Wmeasuregeneral}) are valid for processes applied to classical and quantum states alike. While for a classical ensemble the entropy change, $\Delta S^{\cal P} = S (\eta^{\cal P} ) - S(\rho)$, will be zero this is not so in the general quantum situation, where initial non-diagonal quantum states result in a strictly positive entropy change \cite{NielsenChuang}. We note that while the optimal work values are in principle attainable, practical implementations may be suboptimal resulting in a reduced work gain or a higher work cost. 
The physical meaning of $\Delta S^{\cal P}$ can be grasped by considering a lower bound \cite{Wolf13} on it, $\Delta S^{\cal P} \ge \frac{1}{2} \,  \left\|\rho - \frac{\mathbbm{1}}{d}\right\|_2^2 \, \Delta A^{\cal P}$,  see Appendix~\ref{app:entropybound}. Here $d$ is the dimension of the system and $ \left\| \, \cdot \, \right\|_2$ denotes the Hilbert-Schmidt norm. The first factor quantifies the distance of the initial state from the fully mixed state, while the second factor, $\Delta A^{\cal P}$, quantifies the angle between the diagonal basis of $\rho$ and the projection basis $\{ \Pi^{\cal P}_k \}_k$. These terms correspond to incoherent and coherent mixing contributions. The entropy change is non-trivially bounded only if the initial state is not an incoherent mixture with respect to that basis. The entropy bound is the largest for pure initial states whose basis is mutually unbiased with respect to $\{ \Pi^{\cal P}_k \}_k$. In this case the optimal entropy change is $\Delta S^{\cal P}  = k_B \, T \, \ln d$.

One may  wonder where the work has gone to. There are two equivalent approaches to the accounting of work.  In the present analysis the focus is on the work that the system exchanges, as done in statistical physics \cite{bioexperiments,TLH07,CTH09,Berut12,OSaira12}. In this approach it is often not explicitly mentioned where the work goes to, but the only place work can go to are the externally controlled energy sources. Similarly, the heat, i.e. the energy change minus the work, is established implicitly. For example, in the experimental realisation of classical Landauer erasure with a colloidal silica bead trapped in an optical tweezer \cite{Berut12}, the dissipated heat of erasure was calculated by knowing the applied tilting forces and integrating over the bead's dynamics. The second approach is to collect work in a separate work storage system \cite{SSP14}, as illustrated by the weight in Fig.~\ref{fig:overviewpic} and detailed in Appendix~\ref{app:workstorage}. Both the implicit and the explicit treatment of work are equivalent in the sense that the results obtained in one approach can be translated into the other. 

The thermodynamic assumptions made to prove Eq.~(\ref{eq:Wmeasuregeneral}) are congruent with current literature \cite{Landauer61,Esposito10, AG13, SSP14}; specifically they are:
(T0) an isolated system is a system that only exchanges work and not heat;
(T1) the validity of the \emph{first law} relating the internal energy change, $\Delta U$, of the system during a process to its average heat absorbed and work drawn, $\Delta U = \avg{Q^\abs} - \avg{W}$;
(T2) the validity of the \emph{second law} relating the system's entropy change to its average absorbed heat, $k_B T \, \Delta S \ge \avg{Q^\abs}$, when interacting with a bath at temperature $T$, with equality attainable by an optimal process;
(T3) the thermodynamic entropy to be equal to the von Neumann entropy in equilibrium as well as out-of-equilibrium, $S_{\rm th} = S_{\rm vN}$.
In addition we make the following standard quantum mechanics assumptions:
(Q0) an isolated system evolves unitarily;
(Q1) control of a quantum system includes its coherences.
Details of the proof are in Appendix~\ref{app:mainproof}.
We note that in the single-shot setting whole families of second laws apply \cite{Brandao13b, Lostaglio14} that differ from (T2) stated above. However,  in the limit of infinitely many independent and identically prepared copies of the system these collapse to the standard second law, (T2), on the basis of which Eq.~(\ref{eq:Wmeasuregeneral}) is derived. 

From the information-theory point of view the projections considered here constitute just one example of the larger class of trace-preserving completely positive (TPCP) maps characterising quantum dynamics. Of course, all TPCP maps can be interpreted thermodynamically with the assumptions stated above, resulting in an optimal average work given by a free energy difference. Erasure is another such map whose study forged the link between information theory and thermodynamics. The benefit of discussing ``projections'' here lies in the insight that this focus provides: it uncovers that coherences offer the potential to draw work making it a genuine and testable quantum thermodynamic feature. This work is non-trivial even when the thermodynamic process is operated on the system alone, not involving any side-information \cite{delRio} stored in other degrees of freedom. 

\section{Example} \label{sec:example}

To gain a detailed understanding of thermodynamic projection processes that give the optimal work stated in Eq.~(\ref{eq:Wmeasure}) we now detail one such process for the example of a spin-1/2 particle (qubit), see illustration in Fig.~\ref{fig:Blochpicture}b and \ref{fig:Blochpicture}c as well as Appendix~\ref{app:spinexample}. This process consists of  a unitary evolution, a quasi-static evolution and a quench \cite{AG13}, and it is optimal for any finite-dimensional quantum system as shown in Appendix~\ref{app:general3step}. An experimentalist, Emmy, prepares the spin in a state $\rho = a \pure{0} + (1-a) \pure{1}$ ($a \ge \frac{1}{2}$ w.l.o.g.) exposed to an external magnetic field $\vec{B}\upix 0$ which she controls. The Hamiltonian associated with the system is $H =  - E \,  \pure{e_0} +  E \, \pure{e_1}$ where the energy difference between the aligned ground state, $\ket{e_0}$, and anti-aligned excited state, $\ket{e_1}$, is given by $2 E =  |\vec{\mu}| \, |\vec{B}\upix 0|$ with $\vec{\mu}$ the spin`s magnetic moment. Importantly, in general the spin state's basis, $\{\ket{0}, \ket{1} \}$, are superpositions with respect to the energy eigenbasis, $\ket{0} = \alpha^* \, \ket{e_0} + \beta^* \, \ket{e_1}$ and $\ket{1} = \beta \, \ket{e_0} - \alpha \, \ket{e_1}$ with $|\alpha|^2 +|\beta|^2=1$. For the optimal implementation of the projection $\rho \to \eta^H = \sum_{k = 0, 1} \pure{e_k} \, \rho \, \pure{e_k}$ Emmy now proceeds with the following three steps. 

Firstly, she isolates the spin from the bath and modifies external magnetic fields to induce a unitary rotation, $V = \ket{e_0} \bra{0} + \ket{e_1} \bra{1}$, of the spin into the energy basis. In nuclear magnetic resonance (NMR) \cite{Brazil} and pulsed electron spin resonance (ESR) experiments \cite{Gavin} such rotations are routinely realised by radio-frequency and microwave pulses respectively, as evidenced by Rabi oscillations. The power, duration and phase of such a pulse would be chosen to generate the spin-rotation along the green circle until the desired unitary $V$ is achieved. In the same step Emmy adjusts the strength of the external B-field such that the spin state $\rho_1 = V \, \rho \, V^{\dag}$ is Boltzmann-distributed at temperature $T$ with respect to the energy gap of the Hamiltonian at the end of the step, $H\upix 1$. In NMR or ESR the B-field magnitude is tuned quickly on the $T_1$ timescale to achieve the desired energy gap. 
In the second step, Emmy wants to implement a quasi-static evolution of the spin that is now thermal. She brings the spin in contact with the heat bath at temperature $T$ and quasi-statically adjusts the magnitude of the external B-field allowing the spin state to thermalise at all times. The final B-field, $\vec{B}\upix 2$, is chosen such that the final thermal state becomes $\eta^H$. In ESR this step can be realised by changing the external B-field slowly on the $T_1$ timescale so that the spin continuously equilibrates with its environment. Finally, Emmy isolates the spin from the environment and quickly changes the B-field to its original magnitude while the state remains $\eta^H$.

During Step 1 and 3 the system was isolated and the average work drawn is thus just the average energy change. During Step 2 the average work is the equilibrium free energy difference between the final and initial thermal states at temperature $T$, see Appendix~\ref{app:spinexample} for details. In NMR/ESR the work contributions drawn from the spin system are done on the external B-field and the microwave mode. This could be detected by measuring the stimulated emission of photons in the microwave mode or observing current changes induced by the spins dynamics \cite{Brazil,Gavin}.
The overall thermodynamic process has now brought the spin from a quantum state with coherences, $\rho$, into a state without coherences, $\eta^H$, while keeping the average energy of the spin constant. The net work drawn during the three steps adds up to $\avg{W} =  k_B T \, (S (\eta^H) - S(\rho))$ showing the attainability of the optimum stated in Eq.~(\ref{eq:Wmeasure}) for the spin-1/2 example. We note that Eq.~(\ref{eq:Wmeasure}) is also the maximal work that can be extracted from a \emph{qubit} state $\rho$ under \emph{any} transformation of the system that conserves its average energy, $U := \tr[H \, \rho]$, i.e. for qubits $\eta^H$ is the optimal final state under this condition. 

We emphasise that this optimal implementation involves a finely tuned and controlled operation that relies on knowledge of the initial state $\rho$. This is akin to the situation considered in \cite{delRio} where knowledge of the initial global state of system and memory is required for optimal erasure with side-information. It is important to distinguish this situation from that of Maxwell demon's who has access to knowledge of the individual micro-states $\Pi^{\cal P}_k$ that make up the ensemble state $\eta^{\cal P}$, and who uses it to beat the second law \cite{Maruyama09}. In the scenario considered here there is no knowledge of the individual micro-states $\Pi^{\cal P}_k$ and the process does not violate the second law, on the contrary, it is derived from it.

\section{Implications} \label{sec:implications}

\subsection{Single-shot analysis} \label{sec:singleshot}

The preceding discussion concerned the \emph{average} work that can be drawn when operating on an ensemble of $N\rightarrow\infty$ independent spins. 
This scenario contrasts with the single shot situation considered in a number of recent publications \cite{Brandao13b, delRio, Aberg13, HO13}.
In particular, two major frameworks \cite{Aberg13, HO13} have recently been put forward to  identify optimal \emph{single-shot} work extraction and work cost of formation in the quantum setting. These frameworks rely on a resource theory approach \cite{Janzing00} and make use of min- and max-relative entropies that originate from one-shot information theory. 
The optimal work extraction schemes of these frameworks require non-diagonal states to be decohered first to become diagonal in the energy basis. This decoherence step is assumed to not have an associated single-shot work. However, the present analysis of energy basis projections  showed that thermodynamic projection processes can yield positive average work, see Eq.~(\ref{eq:Wmeasure}). Therefore one may expect a positive work for removing coherences from a state $\rho$ in the single-shot setting, too. Since our focus is the $N \to \infty$ limit we will not aim to  construct the single-shot case. Nevertheless, to establish a notion of consistency between single-shot results \cite{Aberg13, HO13} and the average analysis presented here we now separate the projection into a diagonal part that can be analysed in the single-shot framework and a non-diagonal part that can be analysed in the average framework.
One possible decomposition of $\rho \to \eta^H$ is the split in three steps each starting and ending with Hamiltonian $H$: $\rho \overset{a}{\to} \rho_1 \overset{b}{\to} \tau^H  \overset{c}{\to} \eta^H$. Here $\rho_1$ is the rotated state defined above and $\tau^H = e^{- H/ k_B T}/ \tr[e^{- H/ k_B T}]$ is the thermal state for the Hamiltonian $H$ at temperature $T$. We can now use a single-shot analysis \cite{HO13} for Steps $b$ and $c$ that involve only states diagonal in the energy basis, giving a single-shot work contribution of $k_BT\,\ln2\,(\Dmin (\rho_1||\tau^H) - \Dmax (\eta^H||\tau^H))$, see Appendix~\ref{app:singleshot}. Here $\Dmin$ and $\Dmax$ are the min- and max-relative quantum entropies, respectively. Taking the limit of $N\to \infty$ copies for Steps $b$ and $c$ and adding the average work contribution for the initial non-diagonal rotation $a$, $\avg{W\upix a} = - \tr[(\rho_1 - \rho) \, H]$, one indeed recovers the optimal average work as stated in Eq.~(\ref{eq:Wmeasure}). 
After making public our results very recently a paper appeared \cite{David15} that derives the work that can be extracted when removing coherences in a single-shot setting. These results are in agreement with Eq.~(\ref{eq:Wmeasure}) and reinforce the above conclusion that coherences are a fundamental feature distinguishing quantum from classical thermodynamics.

\subsection{Quantum fluctuation relations} \label{sec:fluctuationrels}

The key observation was that thermodynamic projection processes can have a non-trivial work and heat. Another instance where this has interesting repercussions is the quantum Jarzynski equality \cite{Mukamel03, TLH07}. This is a generalisation of the prominent classical fluctuation relation valid for general non-equilibrium processes, which has been used to measure the equilibrium free energy surface inside bio-molecules by performing non-equilibrium pulling experiments \cite{bioexperiments}. The quantum version has recently been tested for the first time in a nuclear magnetic resonance experiment \cite{Brazil}.
The quantum Jarzynski relation, $\avg{e^{ {W/ (k_B \, T) }}}  = e^{- {\Delta F / (k_B \, T) }}$, links the fluctuating work, $W$, drawn from a system in individual runs of the same non-equilibrium process, with the free energy difference, $\Delta F$, of the thermal states of the final and initial Hamiltonian,  see Appendix~\ref{app:jarzynski}. In its derivation a system initially in a thermal state $\rho_0$ with respect to Hamiltonian $H\upix 0$ at temperature $T$ is first measured in the energy basis of $H\upix 0$. The Hamiltonian is then varied in time ending in $H\upix {\tau}$ generating a unitary evolution, $V$, of the system, see Fig.~\ref{fig:QuJarzynski}a. A second measurement, in the energy basis of $H\upix {\tau}$, is then performed to establish the final fluctuating energy. For each run the difference of the two measured energies has been associated with the fluctuating work \cite{TLH07}, $\Delta E = - W$. The experiment is repeated, each time producing a fluctuating work value. On average the work extracted from the system during the quantum non-equilibrium process turns out to be $\avg{W^{\rm unitary}} = U (\rho_0) - U(\rho_{\tau})$ where $\rho_{\tau} := V \, \rho_0 \, V^{\dag}$ is the ensemble's state after the unitary evolution, and similarly the average exponentiated work is calculated. 
The above identification $W := - \Delta E$ was made assuming that the system undergoes a unitary process with no heat dissipation. However, the need to acquire knowledge of the system's final energies requires the second measurement. The ensemble state is thus further altered from $\rho_{\tau}$ to $\eta_{\tau}$, the state $\rho_{\tau}$ with any coherences in the energy basis of $H\upix {\tau}$ removed. This step is not unitary - during the projection $\rho_{\tau} \to \eta_{\tau}$ the system may absorb heat, $\avg{Q^{\rm abs}}$, indicated in Fig.~\ref{fig:QuJarzynski}b, whose value depends on \emph{how} the process is conducted. Thus, while the energy difference for the projection is zero, $U(\eta_{\tau}) - U(\rho_{\tau}) =0$, for states $\rho_{\tau}$ with coherences the entropy difference is not trivial, $S(\eta_{\tau}) - S(\rho_{\tau}) = \avg{Q^{\rm abs}_{\rm opt}}/(k_B\,T)  \geq 0$. This implies that in an experimental implementation of the Jarzynski relation the work done by the system on average can be more than previously thought, $\avg{W_{\rm opt}} = \avg{W^{\rm unitary}} + k_B \, T \, (S(\eta_{\tau}) - S(\rho_{\tau}))$. We conclude that the suitability of identifying $W = - \Delta E$, and hence the validity of the quantum Jarzynski \emph{work} relation, depends on the details of the physical process that implements the second measurement. This conclusion is not at odds with previous experiments \cite{Brazil} which showed nature's agreement with $\avg{e^{- {\Delta E / (k_B \, T) }} }  = e^{- {\Delta F / (k_B \, T) }}$, involving the average of the exponentiated measured fluctuating energy. 

\begin{figure}[t]
    \includegraphics[width=0.37\textwidth]{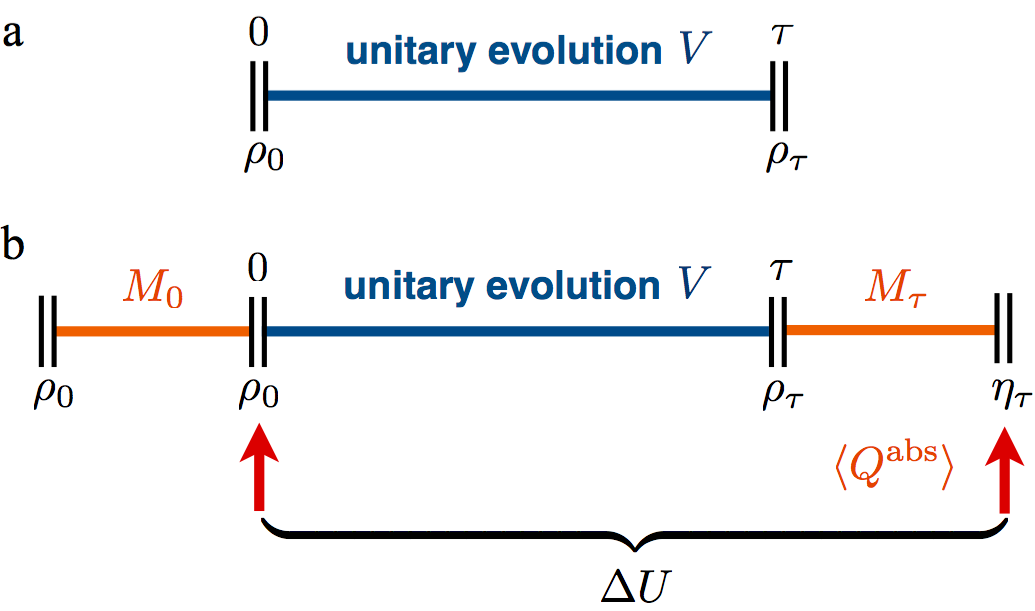}
        \caption{\label{fig:QuJarzynski} 
       {\bfseries Dynamical steps in a quantum fluctuation experiment.}   
      {\bfseries a,} The quantum Jarzynski relation is described as characterising the non-equilibrium work of processes that start in a thermal state $\rho_0$ and evolve unitarily ($V$), driven by a changing Hamiltonian, reaching the final state $\rho_{\tau}$ at time $\tau$. This unitary process has no heat contribution.
        {\bfseries b,}  Illustration of three steps that are assumed in mathematical derivations of the quantum Jarzynski relation \cite{Mukamel03, TLH07}:  initial energy measurement of $H\upix 0$ indicated by $M_0$,  unitary evolution, and final energy measurement of $H\upix {\tau}$ indicated by $M_{\tau}$. The  ensemble state evolves here from $\rho_0$ to $\rho_{\tau}$ and then to $\eta_{\tau}$, the state $\rho_{\tau}$ with its coherences removed. The observed average energy difference $\Delta U = U(\eta_{\tau}) - U(\rho_0)$ encompasses both, the unitary process and the second projection process, and can in general contain a heat contribution $\avg{Q^{\rm abs}}$, in contrast to {\bfseries a}.
}
\end{figure}

\subsection{Correlated systems} \label{sec:corrsystems}

It is insightful to extend the thermodynamic analysis of projections to correlated systems. An experimenter may have access not only to the system $S$ but also the auxiliary systems $A$ with which $S$ is correlated \cite{delRio}. She can then perform a global operation, $\rho^{SA} \to \eta^{SA}$, that implements a projection $\{ \Pi^{\cal P}_k \}_k$ locally on the system $S$, i.e. $\rho^S : = \tr_A[\rho^{SA}] \to \eta^S : = \tr_A[\eta^{SA}] = \sum_k \Pi^{\cal P}_k \, \rho^S \, \Pi^{\cal P}_k$, while leaving the reduced state of the auxiliary system unchanged, i.e. $\rho^A := \tr_S[\rho^{SA}] = \tr_S[\eta^{SA}]$.
By doing so the experimenter can optimally draw the overall work $\avg{{\cal W}_\opt}  = k_B \, T \, \Delta {\cal S}^{\cal P}  - \Delta U^{\cal P}$, where $\Delta {\cal S}^{\cal P}$ is the entropy change for the state of system+auxiliary and $\Delta U^{\cal P}$ is still the energy change of the system \emph{alone}. This quantity can be re-written as  the sum of two terms: $\avg{W_\opt}$, the extractable work when operating on the system alone given in Eq.~(\ref{eq:Wmeasuregeneral}), and $\delta^{\cal P} (A:S)$, a positive term quantifying the quantum correlations between $S$ and $A$, see Appendix~\ref{app:correlations}. 
The latter contribution was previously identified in an inspiring paper by Zurek \cite{Zurek03}. It depends on the choice of projectors and is related to, but broader than, quantum discord \cite{discord} which is optimised over all possible projectors. This means that even states of system and auxiliary that can be considered classically correlated (i.e. no discord) provide an advantage for drawing work  contrasting with the erasure process where this only occurs for highly entangled states \cite{delRio}. The gap between these two sets of correlated states is an intriguing fact and calls for further exploration of the link between thermodynamics and information theory in the quantum regime.

\section{Conclusions} \label{sec:conclusions}

To conclude, erasure is not the only irreversible information processing task -- in the quantum regime a second fundamental process exists that mirrors Landauer's erasure. In contrast to the minimum heat limit of erasure, thermodynamic projection processes have a maximum work limit. While the former is non-zero for the erasure of classical \emph{and} quantum bits,  optimal thermodynamic projection processes have a non-zero work \emph{only} when applied to quantum states with coherences.
The optimal average work stated in Eqs. (\ref{eq:Wmeasure}) and (\ref{eq:Wmeasuregeneral}) constitutes an experimentally accessible quantum thermodynamic prediction. Future experiments testing this optimal work may be pursued with current setups, for instance with NMR/ESR techniques \cite{Brazil,Gavin} or single atoms \cite{Maunz11, Volz11}, and promise to be accessible with other platforms entering the quantum regime, such as single electron boxes \cite{OSaira12}. Experiments will be limited by practical constraints, such as achieving a quasistatic process and obtaining the maximum work for pure states which may require, for instance, very large B-fields. 

The derivation of the optimal work value is mathematically straightforward, just like that of Landauer's principle. The result's significance is that it opens new avenues of thought and provides key input for the construction of a future quantum thermodynamic framework. For example, the developed approach opens the door to investigate the connection between microscopic statistical physics and macroscopic thermodynamics in the quantum regime. While it is straightforward to identify the thermodynamic work of quantum processes involving macroscopic ensembles, what is needed is a microscopic concept of work that when averaged, gives the correct macroscopic work. The microscopic work concept should be valid for general (open) quantum processes and quantum states (including coherences), and only require access to properties of the system. While single-shot approaches have discarded coherences \cite{Aberg13, HO13}, fluctuating work approaches cannot be applied directly to a system undergoing open quantum evolution \cite{CTH09}. 

The observation is also important from the experimental perspective as testing quantum thermodynamic predictions will involve measurement -- a projection process. We have argued that measurements, such as those required in establishing the Jarzynski equality, are not necessarily thermodynamically neutral. Indeed, they can be implemented in different physical ways and in general play an active role in thermodynamics, contributing a non-zero average heat and work. This new perspective gives physical meaning to the change of entropy in the debated quantum measurement process - it provides a capacity to draw work. Specifically, work can be drawn when  \emph{coherences} of a state are removed during an unselective measurement. 

Finally, it is apparent that optimal thermodynamic projection processes require use of knowledge of the initial state $\rho$, i.e. its basis and eigenvalues. One may be inclined to exclude use of such knowledge, particularly when considering projections in the context of measurement which is often associated with the acquisition of knowledge. Such restriction would necessarily affect the set of assumptions (T0-T3, Q0-Q1) in the quantum regime. These could be changed, for example, to the second law not being possible to saturate (cf. T2) or to choosing a new quantum non-equilibrium entropy that only considers the state's diagonal entries (cf. T3). The latter would mean a departure from standard quantum information theory where entropies are basis-independent. Thus whichever approach one takes - not making or making a restriction - quantum coherences will contribute a new dimension to thermodynamics. They either lead to non-classical work extraction or they alter the link between information theory and thermodynamics in the quantum regime. The line drawn here between the assumptions (T0-T3, Q0-Q1) and results (Eqs.~(\ref{eq:Wmeasure}) and (\ref{eq:Wmeasuregeneral})) establishes a frame for this possibility to be investigated.

\acknowledgements

We thank T. Deesuwan, M. Wolf and R. Renner, G. Morley, R. Uzdin and D. Reeb for insightful discussions and J. Gemmer, R.  Renner, S. Horsley and T. Philbin for critical reading of the manuscript. 
P.K. acknowledges support from the Swiss National Science Foundation (through the National Centre of Competence in Research `Quantum Science and Technology') and the European Research Council (grant 258932).
J.A. is supported by the Royal Society and EPSRC. J.A. thanks the Isaac Newton Institute in Cambridge where part of this work was conceived for the stimulating environment and kind hospitality. This work was supported by the European COST network MP1209.

\appendix

\section{Proof of Eq.~(\ref{eq:Wmeasuregeneral})} \label{app:mainproof}

Using the first law (T1) the average work drawn in a thermodynamic projection process $\rho \to \eta^{\cal P}$ is simply $\avg{W} =  \avg{Q^\abs} - \Delta U^{\cal P}$, where $\Delta U^{\cal P}$ is the average energy change for that process. Relating the average heat absorbed by the system during the process to its entropy change one then obtains $\avg{W} \le k_B T \, \Delta S^{\cal P}$ (T2). Here $\Delta S^{\cal P}$ is the difference of von Neumann entropies of the system's state before and after the projection (T3). The average work drawn is thus $\avg{W} \le k_B T \, \Delta S^{\cal P} - \Delta U^{\cal P}$, where the entropy change is non-negative and the energy change can be either positive or negative. The stated \emph{optimal} work, $\avg{W_\opt}$, is achieved when the inequality is saturated by an optimal process (T2) the implementation of which may require knowledge of the initial state and control of coherences (Q1). In the special case of a projection onto the energy eigenbasis $\{ \Pi^H_k\}_k$ the internal energy change is zero, $\Delta U^H = 0$, and one obtains Eq.~(\ref{eq:Wmeasure}).

\section{Spin example} \label{app:spinexample}

We here detail the three steps of the optimal thermodynamic projection process of the spin system discussed in the main text. Emmy starts with a spin in state 
\begin{align}
	\rho = a\,\pure{0} +(1-a)\,\pure{1},
\end{align}
with Blochvector
\begin{align}
	\vec{s}_\rho := \tr[\rho \, \vec{\sigma}] = (2 a -1) \,  \hat{0},
\end{align}
where $\vec{\sigma}$ is the vector of the three Pauli matrices, $\sigma_1, \sigma_2$ and $\sigma_3$ and $\hat{0} = \tr[\pure{0} \, \vec{\sigma}]$ is the unit vector in the Blochsphere pointing from the origin to the state $\ket{0}$, see Fig.~\ref{fig:Blochpicture}b. We assume without loss of generality that $a \ge \frac{1}{2}$. If this was not the case, the labels $\pure{0}$ and $\pure{1}$ should be interchanged. The spin's initial Hamiltonian is given by $H=- E\,\big(\Pi^H_0 - \Pi^H_1\big)$, where $\Pi^H_k = \pure{e_k}$ with $k=0,1$ are the rank-1 projectors onto the two energy eigenstates and $E >0$. This Hamiltonian arises when the spin is exposed to an external magnetic field $\vec{B}^{(0)}$. The energy separation of the aligned ground state, $\ket{e_0}$, and anti-aligned excited state, $\ket{e_1}$, is $2E = 2 \, |\vec{\mu}| \, |\vec{B}^{(0)}|$, where $\vec{\mu}$ is the magnetic moment of the spin. A general initial state $\rho$ is not diagonal in the basis $\{ \ket{e_0}, \ket{e_1}\}$, in other words the spin's eigenstates are superpositions with respect to the energy eigenbasis, $\ket{0} = \alpha^* \, \ket{e_0} + \beta^* \, \ket{e_1}$ and $\ket{1} = \beta \, \ket{e_0} - \alpha \, \ket{e_1}$ with $|\alpha|^2 +|\beta|^2=1$. The spin's Blochvector, $\hat{0}$, is then \emph{not} parallel to the B-field, $\vec{B}^{(0)}$. Emmy wants to obtain the state where the coherences with respect to the energy basis $\{\ket{e_0}, \ket{e_1}\}$ have been removed,
\begin{align}\begin{split}
	\eta^H &= \sum_{k=0,1} \Pi^H_k\, \rho \, \Pi^H_k \\
	&= \sum_{k=0,1} \tr[\Pi^H_k\,\rho ]\,\Pi^H_k = p\ \Pi^H_0 + (1-p)\ \Pi^H_1,
\end{split}\end{align}
where $p=\tr[\Pi^H_0\,\rho]$ is the probability for obtaining outcome $\ket{e_0}$ in a measurement of $H$. Here $\eta^H$ has the same average energy as the initial state $\rho$, 
\begin{align}\begin{split}
	\tr[\eta^H\, H] &= \sum_{k=0,1} \tr[\,\Pi^H_k\, \rho\, \Pi^H_k\, H\,] \\
	&= \sum_{k=0,1} \tr[\,\rho\, \Pi^H_k\, H\, \Pi^H_k\,] = \tr[\rho\, H],
\end{split}\end{align}
where we used orthonormality and completeness of the projectors $\{ \Pi^H_0, \Pi^H_1 \}$. The Blochvector of the final state is defined as 
\begin{align}
	\vec{s}_{\eta} := \tr[\eta^H \, \vec{\sigma}] = (2 p -1) \,  \hat{e}_0 
\end{align}
where $\hat{e}_0 = \tr[\pure{e_0} \, \vec{\sigma}]$ is the unit vector in the Blochsphere pointing from the origin to the state $\ket{e_0}$. Since geometrically the mapping  $\rho \rightarrow \eta^H$ is a projection of $\vec{s}_\rho$ onto the vertical axis in the Blochsphere, the length of the final Blochvector, $\vec{s}_{\eta}$, is shorter than the initial Blochvector, $\vec{s}_\rho$. This shortening is associated with an entropy increase \cite{NielsenChuang}. When describing the process in the following we assume that $p\geq\frac{1}{2}$ in accordance with the illustration in Fig.~\ref{fig:Blochpicture}b. At the end of this section we come back to the case $p<\frac{1}{2}$.

Emmy proceeds with three steps made up of quantum thermodynamic primitives with known work and heat contributions \cite{AG13}, Fig.~\ref{fig:Blochpicture}b:
\begin{align*}
	(\rho, H) \stackrel{1}{\longrightarrow} (\rho_1, H\upix 1) \stackrel{2}{\longrightarrow} (\eta^H, H\upix 2)\stackrel{3}{\longrightarrow} (\eta^H, H).
\end{align*}
In the first step, $(\rho, H) \stackrel{1}{\longrightarrow} (\rho_1, H\upix 1)$, Emmy isolates the spin from the bath and rotates the B-field such that the variation of the field induces a unitary transformation of the spin into the energy eigenbasis, with unitary $V = \ket{e_0}\bra{0} + \ket{e_1}\bra{1}$. The state after this step is 
\begin{align} \label{eq:rho1}
	\rho_1 = V\,\rho\, V^\dagger = a\ \Pi^H_0 +(1-a)\ \Pi^H_1.
\end{align}
The B-field after this step, $\vec{B}\upix 1$, is chosen such that the new Hamiltonian $H\upix 1 = - E\upix 1\,\big( \Pi^H_0 - \Pi^H_1\big)$ has eigenvalues $E\upix 1 = |\vec{\mu}| \, |\vec{B}\upix 1| = \frac{k_B \,T}{2}\,\ln\frac{a}{1-a}$, where $k_B$ is the Boltzmann constant and $T$ is the temperature of the heat bath that Emmy will use in the next step. This choice of the B-field makes the state $\rho_1$ a thermal state with respect to $H\upix 1$ at temperature $T$, i.e. $\rho_1 = \frac{e^{-\beta H\upix 1}}{Z\upix 1}$ with $Z\upix 1 = \tr\big[e^{-\beta H\upix 1}\big]$ and inverse temperature $\beta = \frac{1}{k_B \, T}$. Since the system was isolated in the first step no heat exchange was possible and the entire average energy change of the system is drawn from the system as work $\avg{W\upix 1}  = - \tr[\rho_1 \, H\upix 1 - \rho \, H]$.  Physical constraints may make this process difficult to realise, for instance, pure initial states would require a B-field, $B\upix 1$, of infinite magnitude because thermal states at any finite temperature are only pure if the energy gap is infinite. In this case there is a trade-off between the maximal magnitude the B-field can reach and the precision with which the process is carried out. In the following we assume that the maximal B-field is large enough to make the error in the precision negligibly small.

In the second step, $(\rho_1, H\upix 1) \stackrel{2}{\longrightarrow} (\eta^H, H\upix 2)$, Emmy brings the spin in contact with the bath at temperature $T$, not affecting the spin's state as it is already thermal. She then quasi-statically decreases the magnitude of the B-field, while keeping the system in contact with the bath at all times, such that the final Hamiltonian is $H\upix 2 = -E\upix 2\,\big( \Pi^H_0 - \Pi^H_1\big)$ where the B-field is chosen such that $E\upix 2 = \frac{k_B \,T}{2}\,\ln\frac{p}{1-p}$ where $p$ is the probability of measuring $-E$ in the initial state, $\rho$. The quasi-static evolution means that the system is thermalised at all times, arriving in the final state 
\begin{align}
	\frac{e^{-\beta H\upix 2}}{Z\upix 2} = p\ \Pi^H_0 + (1-p)\ \Pi^H_1 = \eta^H
\end{align}
which is thermal with respect to $H\upix 2$ where $Z\upix 2 = \tr\big[e^{-\beta H\upix 2}\big]$. This state is exactly $\eta^H$, the desired final state after the projection. The quasi-static process considered here has a known average work given by the free energy difference \cite{Gemmer, AG13, Aberg13}, $\avg{W\upix 2} = - \big(\, F^{(2)}_{T}(\eta^H) - F^{(1)}_{T}(\rho_1) \, \big)$ where $F^{(2)}_{T}(\eta^H) = \tr[\eta^H \,H\upix 2] - k_B\, T\, S(\eta^H)$ and $F^{(1)}_{T}(\rho_1) = \tr[\rho_1 \,H\upix 1] - k_B\, T\, S(\rho_1)$ are standard thermal equilibrium free energies and $S$ is the von Neumann entropy defined by $S(\eta^H) = -\tr[\eta^H\, \log \eta^H]$ and likewise for $\rho_1$.\\

Finally, in the third step, $(\eta^H, H\upix 2)\stackrel{3}{\longrightarrow} (\eta^H, H)$, Emmy isolates the spin from the bath and changes the energy levels of the Hamiltonian such that it becomes the initial Hamiltonian $H$ again. This step is done quickly so that the state of the spin does not change. Because the system is isolated the energy change in this step is entirely due to work $\avg{W\upix 3} = - \tr[\eta^H \, (H - H\upix 2)]$.
In total, this thermodynamic process has brought the spin from the quantum state $(\rho,H)$ to the state $(\eta^H,H)$ while not changing the energy of the spin, $ \tr[(\rho-\eta^H)\,H] = 0$. The overall average work drawn from the spin is 
\begin{align}\begin{split} \label{eq:spinWext}
	\avg{W} &= \avg{W\upix 1} + \avg{W\upix 2} + \avg{W\upix 3}  \\
	&= -\tr[\rho_1\,H\upix 1] + \tr[\rho\,H] -\tr[\eta^H\,H\upix 2] + k_BT\,S(\eta^H) \\
	&\quad \, + \tr[\rho_1\, H\upix 1] -k_BT\, S(\rho_1) -\tr[\eta^H\,(H-H\upix 2)] \nonumber \\
	&= k_BT\, (S(\eta^H) - S(\rho_1)) \nonumber \\
	&= k_BT\, (S(\eta^H) - S(\rho)) \equiv \avg{W_\opt}, \nonumber
\end{split} \end{align}
showing the optimality of the three step process for the spin example,  cf. Eq.~(\ref{eq:Wmeasure}). 

The above example assumed $p\geq\frac{1}{2}$. Suppose now that the probability to find the final state $\eta^H$ in the ground state $\ket{e_0}$ with respect to the Hamiltonian $H$ was smaller than to find it in the excited state $\ket{e_1}$, i.e. $p<\frac{1}{2}$. Proceeding through the three steps described one finds that the mathematics is exactly the same. In particular, after Step 2 $\eta^H$ is a thermal state with respect to $H \upix 2$ at inverse temperature $\beta$.
The only difference occurs in the interpretation as for the Hamiltonian $H\upix 2$ the ground state is $\ket{e_1}$ because $E \upix 2 = \frac{k_B \, T}{2}\ln \frac{p}{1-p} < 0$ is negative. This is feasible by making the B-field $B \upix 2$ negative, thus swapping the ground and the excited state. Consequently the analysis  above and the resulting expression of the total extracted work remain the same.\\

The work extracted in the individual steps of the thermodynamic projection process can be either positive or negative, depending on the initial state $\rho$, the Hamiltonian $H$ and the temperature $T$ of the heat bath. Their sum, $\avg{W}$, is strictly positive whenever the initial state was not diagonal in the energy eigenbasis, a consequence of the entropy increase \cite{NielsenChuang} from $\rho$ to $\eta^H$. On the other hand for classical states -- all diagonal in the energy basis -- the optimal work for such a projection is always zero. 
Appendix~\ref{app:general3step} extends the optimality proof of the above three step process to the general finite-dimensional case.\\

\textbf{A note on optimal work extraction at constant average energy.} 
Assume we are given an initial state $\rho$ and a non-degenerate Hamiltonian $H$ for a quantum system. The goal is to find the maximal work that can be obtained in a thermodynamic process that involves a heat bath at temperature $T$ under the restriction that the average energy of the system after the process is the same as it was before the process, $U:=\tr[\rho \, H]$. Using Eq.~(\ref{eq:Wmeasuregeneral}) together with the condition that internal energy does not change this amounts to finding the maximum over the set of states $\sigma$ with $\tr[\sigma \, H] = U$, 
\begin{align}\begin{split}
	\avg{W^{\max}} &:= \max_{\sigma} k_B \, T (S(\sigma) - S(\rho)) \\
	&= k_B \, T (\max_{\sigma} S(\sigma) - S(\rho)).
\end{split}\end{align}
It is well-known that at a fixed expectation value of an observable $H$ the Gibbs states $\sigma_{\lambda} = {e^{- \lambda H} / \tr[e^{- \lambda H} ]}$  are the states of maximal entropy \cite{PuszWoronowicz, Lenard}. Here the parameter $\lambda$ has to be chosen such that the energy of the Gibbs state matches $U$ - therefore there is only one $\sigma_{\lambda^*}$, with $\lambda^*$ such that $\tr[ \sigma_{\lambda^*} \, H] \equiv U$, that gives the maximum here. The maximum entropy is then
\begin{align}\begin{split}
S(\sigma_{\lambda^*}) &= - \tr[ \sigma_{\lambda^*} \ln \sigma_{\lambda^*}] = - \tr[ \sigma_{\lambda^*} (- \lambda^* H - \ln \tr[e^{- \lambda^* H} ]) ]\\
&= \lambda^* \, U + \ln \tr[e^{- \lambda^* H} ]
\end{split}\end{align}
and the maximum average work that can be extracted from $\rho$ at fixed average energy $U$ is then
\begin{align}
	\avg{W^{\max}} =  k_B \, T ( \lambda^* \, U + \ln \tr[e^{- \lambda^* H}] - S(\rho)) ,
\end{align}

For the special case that the system is a qubit (two-dimensional) the optimum Gibbs state for work extraction $\sigma_{\lambda^*}$ is identical to the projected state $\eta^H = \sum_{k=0,1} \Pi\upix k \, \rho \, \Pi\upix k$ and the maximal work that can be drawn from a system starting in state $\rho$, while keeping its average energy fixed, is $\avg{W_{\opt}}$ in Eq.~(\ref{eq:Wmeasure}). To see this we expand $H = E\upix 0 \, \Pi\upix 0 +  E\upix 1 \, \Pi\upix 1$ and $U = p \, E\upix 0 + (1-p) E\upix 1$ with $p:= \tr[\rho \, \Pi\upix 0]$. Now here $\lambda^*$ must be chosen such that $\sigma_{\lambda^*} = {e^{- \lambda^* H} / \tr[e^{- \lambda^* H} ]} = p \, \Pi\upix 0 +  (1 - p) \, \Pi\upix 1$, i.e. $p= e^{- \lambda^* E\upix 0} / (e^{- \lambda^* E\upix 0}  + e^{- \lambda^* E\upix 1})$, so that $\sigma_{\lambda^*}$ has just the right energy $\tr[\sigma_{\lambda^*}  \, H] = U$. On the other hand the projection state has the same expansion, $\eta^H = p \, \Pi\upix 0 +  (1 - p) \, \Pi\upix 1 = \sigma_{\lambda^*}$. We note that this coincidence is not true for higher dimensional systems where the energy-projected state $\eta^H$ will in general have a non-monotonous, non-canonical distribution in its energy eigenbasis, while $\sigma_{\lambda^*}$ must be Gibbs-distributed.

Considering the illustration in Fig.~\ref{fig:Blochpicture}b, the qubit states $\sigma$ fulfilling the condition $\tr[\sigma\,H]$ are located on the plane which contains $\rho$ and is perpendicular to the $\ket{e_0}$-$\ket{e_1}$-axis. On the other hand, in the Bloch picture a state has higher entropy the closer it is to the center of the sphere. Hence, the optimal final state when extracting work from $\rho$ while conserving the average energy of the system is the state $\rho$ projected to the $\ket{e_0}$-$\ket{e_1}$-axis, i.e. $\eta^H$.

\section{Work storage system} \label{app:workstorage}

In the previous section it was stated that work can be drawn from a quantum system when undergoing a thermodynamic projection process. But where has the work gone to?

There are two approaches of accounting for work that are mirror images to each other. One approach \cite{Aberg13,MT11,Brazil,TLH07,Gemmer,Esposito10,AG13,Seifert,OSaira12} focusses on the work that the system exchanges, as described above. Here it is often not explicitly mentioned where the work goes to, but the only place it can go to are the externally controlled energy sources, see Fig.~\ref{fig:overviewpic}. 
Another way of accounting is to explicitly introduce a work system to store the work drawn \cite{HO13,SSP14}. One way of doing so in an average scenario is to introduce \cite{SSP14} a `suspended weight on a string', described by a quantum system $W$, that could be raised or lowered to store work or draw work from it. Specifically, the Hamiltonian of the work storage system is defined as $H^W=m \, g \, x$, representing the energy of a weight of mass $m$ in the gravitational field with acceleration $g$ at height $x$. In addition, an explicit thermal bath $B$ is introduced \cite{PuszWoronowicz, Lenard} consisting of a separate quantum system in a thermal (or Gibbs) state $\tau^B$. Both, the explicit work storage system and the heat bath are illustrated in Fig.~\ref{fig:overviewpic}.
In the latter approach the total system starts in a product state of system $S$ (e.g. spin), bath $B$, and weight $W$, $\rho^{SBW} = \rho^S\otimes \tau^B \otimes \omega^W$, which together undergo \emph{average energy conserving unitary evolution} with $V$:
\begin{align}
\rho^{SBW} \longmapsto\ \sigma^{SBW} = V\, \rho^{SBW} V^\dagger.
\end{align}
The assumption is that the total Hamiltonian is the sum of local terms, $H^{SBW} = H^S + H^B + H^W$. The average energy conservation constraint then reads $\tr[\,(\sigma^{SBW} - \rho^{SBW})\, H^{SBW}] = 0$ and the average work extracted to the work storage system, $\avg{W}$, is identified with
\begin{align}
	\tr[\, (\sigma^{SBW} - \rho^{SBW})\, H^W \,] = \avg{W}.
\end{align}
Both the implicit and the explicit treatment of work are equivalent in the sense that the results obtained in one language can be translated in the other and vice versa. In particular, the implicit description used in this text \cite{AG13} has an equivalent explicit formulation \cite{SSP14}.

In the next section we will discuss single-shot extractable work in a projection process. 
One possibility to define work in this context is to chose the explicit work storage system as a `work qubit' with a specific energy gap which has to be in a pure energy eigenstate before and after the protocol \cite{HO13}. This way it is guaranteed that full knowledge about its state is present at all times and the work is stored in an ordered form. In this scenario the allowed unitary operations $V$ on the whole system $SBW$ have to conserve the energy exactly, not only on average, which amounts to $[V, H^{SBW}]=0$.

\section{General three step process} \label{app:general3step}

It is straightforward to generalise the proof of optimality from the two-dimensional spin-1/2 example to thermodynamic projection processes in dimension $d$. Again the projectors $\{ \Pi^H_k \}_k$ map onto the energy eigenspaces of the Hamiltonian, $H = \sum_{k} E\upix 0_k \, \Pi^H_k$, where $E\upix 0_k$, $k=1,\dots,d$, are the energy eigenvalues. A general initial state can be written as $\rho = \sum_{j=1}^d a_j \, \pure{j}$ where $a_j \ge 0$ are probabilities, $\sum_{j=1}^d a_j = 1$, $\pure{j}$ are rank-1 projectors on the corresponding eigenvectors $\ket{j}$, and $j=1,\dots,d$. 
A unitary operation, $V$, is now chosen such that it brings the initial configuration $(\rho, H)$ into the new diagonal and thermal configuration $(\rho_1 = V \, \rho \, V^{\dag}, H\upix 1)$ where $\rho_1 = \sum_{k} a_k \, \Pi^H_k$ and $H\upix 1 = \sum_{k} E\upix 1_k \, \Pi^H_k$. The new energy eigenvalues, $E\upix 1_k$, are adjusted such that the probabilities $a_k$ are thermally distributed with respect to $H\upix 1$ for the bath temperature $T$.  Adjusting the Hamiltonian eigenvalues while letting the state thermalise at all times now results in a isothermal quasi-static operation from $(\rho_1, H\upix 1)$ to $(\eta^H = \sum_{k} p_k \, \Pi^H_k, H\upix 2 = \sum_{k} E\upix 2_k \, \Pi^H_k)$. Here the new energy eigenvalues, $E\upix 2_k$, are chosen to be thermal (at $T$) for the state's probabilities which are given by $p_k = \tr[\rho \, \Pi^H_k]$. Finally, a quench brings the thermal configuration $(\eta^H, H\upix 2)$ quickly into the non-equilibrium state $(\eta^H, H)$. 
The average work for this overall process is $\avg{W} = \sum_{j=1}^3 \avg{W\upix j}$ 
where $\avg{W\upix 1} = - \tr[H\upix 1 \, \rho_1 - H \, \rho]$ and $\avg{W\upix 3} = - \tr[H \, \eta^H - H\upix 2 \, \eta^H]$ because the first and third steps are unitary (Q0+T0). The quasistatic step's work is \cite{AG13, Aberg13}  $\avg{W\upix 2} = - F\upix 2_{T} + F\upix 1_{T}$ where $F\upix 1_{T} = \tr[H\upix 1 \, \rho_1] - k_B \, T \, S(\rho_1)$ is the thermal equilibrium free energy for Hamiltonian $H\upix 1$, and similarly, $F\upix 2_{T} = \tr[H\upix 2 \, \eta^H] - k_B \, T \, S(\eta^H)$. Summing up and using $\tr[H \, (\rho - \eta^H)] = 0$, one obtains 
\begin{align}
\avg{W} = k_B \, T (S(\eta^H) - S(\rho)) \equiv \avg{W_\opt} ,
\end{align} 
concluding the optimality proof of the  process sequence.

\section{Lower bound on entropy change} \label{app:entropybound}

The entropy change during a projection with projectors $\{ \Pi^{\cal P}_k=\pure{\phi_k}\}_k$ can be lower bounded. In the following, $\| B \|_2 = \sqrt{\tr[B^\dagger B]}$ denotes the Hilbert-Schmidt norm of a linear operator $B$ acting on a $d$-dimensional Hilbert space describing the quantum system of interest. The lower bound reads \cite{acknowledgeWolf}
\begin{align}\label{eq:lowerbound}
	\Delta S^{\cal P} := S(\eta^{\cal P}) -S(\rho) \geq \frac{1}{2} \, \left\| \rho - \frac{\mathbbm{1}}{d}\right\|_2^2 \, \Delta A^{\cal P}.
\end{align}
Here, $S$ is the von Neumann entropy, $\rho$ the initial state and $\eta^{\cal P} = \sum_k \Pi^{\cal P}_k\, \rho\, \Pi^{\cal P}_k$ the final state after the projection process. Furthermore, $\Delta A^{\cal P}$ is the second smallest eigenvalue of the matrix $\mathbbm{1} - M^TM$ where $M$ is the doubly stochastic matrix given by the entries $M_{kl} = |\avg{\phi_k |\, l\,}|^2$ and $\{\ket{\,l\,}\}_l$ is the eigenbasis of the initial state $\rho$. 

Considering the two main terms on the right hand side of Eq.~(\ref{eq:lowerbound}) separately, $\| \rho - \mathbbm{1}/d \|_2 ^2$ and $\Delta A^{\cal P}$, it becomes apparent that they characterise different properties of the initial state. The first term measures the distance of $\rho$ to the fully mixed state, $\mathbbm{1}/d$, and quantifies the purity of $\rho$. It is maximal for all pure initial states and zero if and only if $\rho = \mathbbm{1}/d$. In the special case of a spin-1/2 system it can be directly related to the length of the Bloch vector describing $\rho$ in the Bloch representation, a link that will be established below. The second term, $\Delta A^{\cal P}$, is related to the overlap of the eigenbasis of $\rho$, $\{\ket{\,l\,}\}_l$, and the projective basis, $\{\ket{\phi_k}\}_k$. It is zero if they are the same and maximal if they are mutually unbiased \cite{Appleby, Durt}. This can be seen as follows: if the two bases are the same, then  the matrix $M$ is a permutation and consequently $M^TM$ is the identity. In this case, $\mathbbm{1}-M^TM$ is the zero matrix and thus $\Delta A^{\cal P} = 0$. If $\{\ket{\phi_k}\}_k$ and $\{\ket{\,l\,}\}_l$ are mutually unbiased, i.e. if they fulfil $|\avg{\phi_k |\, l\,}|^2 = 1/d$ for all $k,l$, the matrix $M$ and thus also $M^TM$ is a rank-1 projector onto the space spanned by the vector $(1,\dots,1)^T$. Hence, $\mathbbm{1}-M^TM$ has eigenvalues $\{0,1,\dots,1\}$. One finds that the second largest eigenvalue is $\Delta A^{\cal P} = 1$, which is also the maximal eigenvalue the matrix $\mathbbm{1}-M^TM$ can have \cite{Ando}.

In the special case of the spin-1/2 system shown in Fig.~\ref{fig:Blochpicture}b, the bound reads $\Delta S^H (\rho) \ge \frac{1}{4} \, \left| \vec{s}_\rho \right|^2 \, \sin^2\theta$, where $\vec{s}_\rho$ is the Bloch vector of the initial state and $\theta$ is the angle between the eigenbasis of $\rho$, $\{\ket{0},\ket{1}\}$, and the projective energy basis, $\{\ket{e_0},\ket{e_1}\}$. Let $\rho = a\,\pure{0} + (1-a)\,\pure{1}$ be the initial state of the qubit. Furthermore, let $\eta^H = p\,\pure{e_0} +(1-p)\,\pure{e_1}$ be the final state after the energy projection, where $p = \tr[\,\pure{e_0} \, \rho\,]$ is the probability to obtain $\ket{e_0}$. As argued in Appendix~\ref{app:spinexample} w.l.o.g. we can assume that $a\geq\frac{1}{2}, p\geq\frac{1}{2}$. In the Bloch representation one can write $\rho = \frac{1}{2}(\mathbbm{1} + \vec{s} \cdot \vec{\sigma})$ and $\eta^H = \frac{1}{2}(\mathbbm{1} + \vec{t} \cdot \vec{\sigma})$. Here we used a different notation for the Bloch vectors of $\rho$, $\vec{s}:=\vec{s}_\rho$, and $\eta^H$, $\vec{t}:=\vec{s}_{\eta}$, for readability.
The Pauli matrices are self-adjoint and fulfil $\tr [\sigma_i \sigma_j] = 2\delta_{ij}$. Hence we find
\begin{align}\begin{split}
\left\|\rho-\frac{\mathbbm{1}}{2}\right\|_2^2 &= \tr\left[ \left(\frac{1}{2}\vec{s}\cdot\vec{\sigma}\right)\left(\frac{1}{2}\vec{s}\cdot\vec{\sigma}\right) \right]\\
&= \sum_{i,j=1}^3 \frac{1}{4}s_{i}s_{j}\tr[\sigma_i\sigma_j] = \frac{1}{2} |\vec{s}|^2,
\end{split}\end{align}
where $|\cdot|$ is the Euclidean metric in $\mathbb{R}^3$. This proves the form of the first factor in the bound. For the factor $\Delta A^H$ notice that by assumption $a\geq\frac{1}{2}, p\geq\frac{1}{2}$ and thus we can write $\pure{e_0} = \frac{1}{2}(\mathbbm{1} + \frac{\vec{t}}{|\vec{t}|}\cdot\vec{\sigma})$ and $\pure{0} = \frac{1}{2}(\mathbbm{1} + \frac{\vec{s}}{|\vec{s}|}\cdot\vec{\sigma})$. Therefore
\begin{align}\begin{split}
\label{eq:overlap}
\ephi &= \tr[\pure{0}\pure{e_0}] \\
&= \frac{1}{4}\tr\left[\mathbbm{1} + \frac{s_{i}}{|\vec{s}|} \sigma_i + \frac{t_i}{{|\vec{t}}|}\sigma_i + \frac{s_{i}t_j}{|\vec{s}||\vec{t}|}\sigma_i\sigma_j\right] \\
&= \frac{1}{4}\left[2 + 0 + 0 + 2\frac{\vec{s}\cdot\vec{t}}{|\vec{s}||\vec{t}|}\right] = \frac{1}{2}[1+ \cos\theta] ,
\end{split}\end{align}
where summations over double indeces are assumed and  $\theta$ is the angle between the two Bloch vectors of $\rho$ and $\eta^H$, see Fig.~\ref{fig:Blochpicture}b. By using orthonormality of both bases $\{\ket{l}\}_l$ and $\{\ket{e_k}\}_k$ it is shown that in the qubit case the matrices $M$  and $M^TM$ have the form 
\begin{align}
M = \begin{pmatrix}
\ephi & |\langle e_0|1\rangle|^2\\
|\langle e_1|0\rangle|^2 & |\langle e_1|1\rangle|^2
\end{pmatrix}
= \frac{1}{2}\begin{pmatrix}
1+\cos\theta & 1-\cos\theta \\
1-\cos\theta & 1+\cos\theta
\end{pmatrix}
\end{align}
and
\begin{align}
M^TM = \frac{1}{2}\begin{pmatrix}
1+\cos^2\theta & 1-\cos^2\theta \\
1-\cos^2\theta & 1+\cos^2\theta
\end{pmatrix}.
\end{align}
Computing the eigenvalues $\{\lambda_1, \lambda_2\}$ of $M^TM$ yields $\lambda_1=1$ and  $\lambda_2 = \cos^2\theta$. Therefore $\Delta A^H = 1-\lambda_2=1-\cos^2\theta = \sin^2\theta$. In total, this concludes the proof of the bound for the special case of a spin-1/2 system when projected in the energy eigenbasis $\{ \Pi^H_k \}_k = \{ \ket{e_0}, \ket{e_1} \}$, 
\begin{align}
\Delta S^H  (\rho) \ge \frac{1}{4} \, \left| \vec{s}_\rho \right|^2 \, \sin^2\theta.
\end{align}

To further illustrate the bound consider the special case when the initial state $\rho$ is pure and its eigenbasis mutually unbiased with respect to the energy eigenbasis, $\{ \ket{e_0}, \ket{e_1} \}$. In this case the final state after the projection, $\eta^H$, is maximally mixed and we find
\begin{align}
\Delta S^H (\rho) = S(\eta^H) - S(\rho) = \ln 2 - 0 \approx 0.69\,.
\end{align}
Here, the lower bound is equal to $\frac{1}{4} = 0.25$ because $\left| \vec{s}_\rho \right| = 1$ for a pure state $\rho$ and $\sin^2\theta = 1$ for mutually unbiased bases. Thus in this example the bound is not particularly tight.

\section{Single-shot analysis} \label{app:singleshot}

Instead of performing a thermodynamic process on an ensemble of $N$ identical and independent copies one can consider a single run of the process. Two major recent frameworks \cite{Aberg13, HO13} have been developed to describe the optimal work that can be drawn from a system in a single run. The proposal by {\AA}berg \cite{Aberg13}, involves changes of the Hamiltonian and identifies work with the deterministic energy change of the system when undergoing a unitary process. The proposal by Horodecki-Oppenheim \cite{HO13}, is formulated in terms of \textit{thermal operations} \cite{Brandao13a}, where work is associated with raising a two-level system, called the `work qubit', with energy gap $W$ deterministically from the ground to the excited state. 

However, when attempting to apply these two frameworks to find the single-shot work for the energy projections $\rho \to \eta^H$ captured by Eq.~(\ref{eq:Wmeasure}) one encounters an obstacle: both frameworks only apply to processes between initial and final states that are classical, i.e. states that are diagonal in the energy basis. {\AA}berg discusses coherences in a separate framework \cite{Aberg14}, which does however not cover single-shot work extraction and only focusses on average quantities, similar to those in other references \cite{SSP14, AG13}. Horodecki-Oppenheim suggest that quantum states with coherences with respect to the energy eigenbasis are first decohered before applying the single-shot protocol. As discussed, apart from decohering there are other thermodynamic projection processes that map the initial state with coherences, $\rho$, to the final state $\eta^H = \sum_k \Pi^H_k \, \rho \, \Pi^H_k$ without coherences, where $\Pi^H_k$ are the projectors on the energy eigenstates of the Hamiltonian, $H$. Eq.~(\ref{eq:Wmeasure}) shows that the  average work extracted in an optimal thermodynamic projection process is strictly positive while the decoherence process has zero work. Therefore one may expect a positive optimal work for projections also in the single-shot setting, with decohering a suboptimal choice, see Fig.~\ref{fig:singleshotmmt}.

\begin{figure}[t]
    \includegraphics[width=0.45\textwidth]{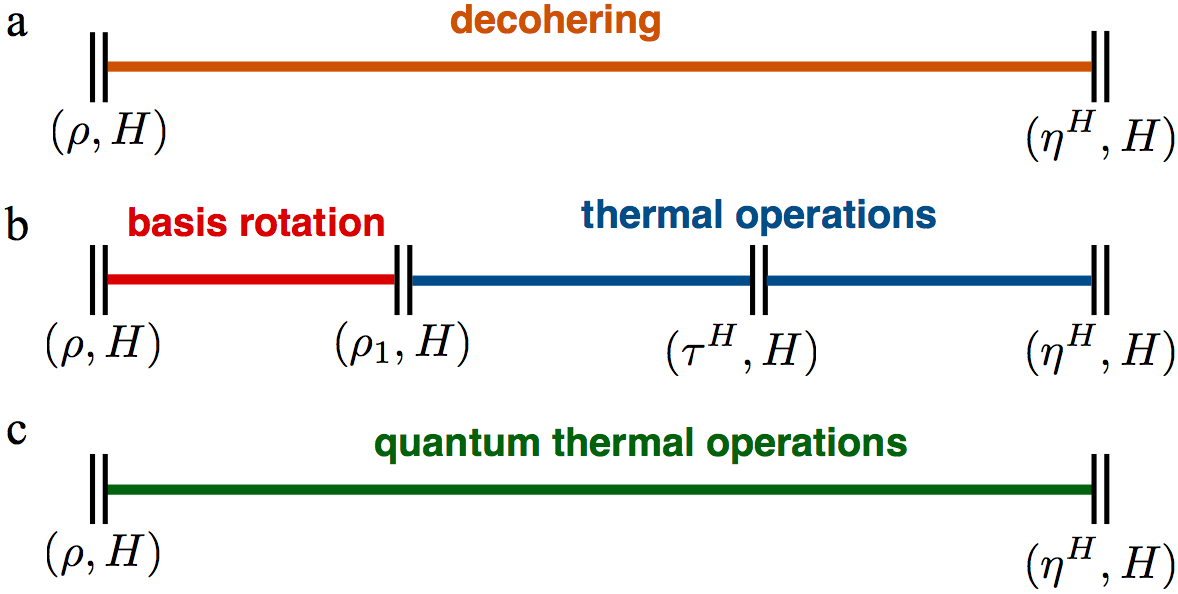}
        \caption{\label{fig:singleshotmmt}
        \footnotesize {\bfseries  Illustration of possible quantum thermodynamic processes that transform $(\rho, H)$ to $(\eta^H, H)$.} {\bfseries a,} Decohering the state in the energy basis extracts no work. {\bfseries b,} 
        To perform a consistency check between the average and single-shot results it is possible to split the process into a basis rotation to $(\rho_1, H)$ with unknown single-shot work, but known average work, and two thermal operations that pass through the thermal state $(\tau^H, H)$ and are treatable in the single-shot framework \cite{HO13}. 
        {\bfseries c,} General quantum thermodynamic processes could allow coherences and need not pass through intermediate fixed states. 
}
\end{figure}

Since our focus here is the $N\rightarrow\infty$ limit we will not aim to construct the single-shot case. Instead, to establish a notion of consistency between the average analysis and previous single-shot work results we consider the sequence $(\rho,H) \stackrel{a}{\longrightarrow} (\rho_1,H) \stackrel{b}{\longrightarrow}  (\tau^H,H) \stackrel{c}{\longrightarrow}  (\eta^H,H)$ in which the Hamiltonian before and after each step are the same and $\rho_1$ is the rotated state defined above, Eq.~(\ref{eq:rho1}).
Here  $\tau^H = \frac{e^{-\beta H}}{Z}$ is the thermal state for Hamiltonian $H$ at inverse temperature $\beta = \frac{1}{k_BT}$ and  $Z=\tr[e^{-\beta H}]$ is the partition function. Step $a$ of this sequence rotates the initial non-diagonal state $\rho$ to the diagonal state $\rho_1$. As discussed, it cannot be treated with the single-shot framework \cite{Aberg13, HO13}  but it is possible to associate an average extracted work with this unitary process, $\avg{W\upix a} = \tr[(\rho-\rho_1)\,H]$. A single-shot analysis according to Horodecki-Oppenheim \cite{HO13} can then be performed for the diagonal steps $b$ and $c$. This is possible because the steps go via the thermal state $\tau^H$. Step $b$ brings $\rho_1$ to $\tau^H$ and allows the extraction of the single-shot work \cite{HO13}
\begin{align}
	W\upix b = k_B\, T\,\ln2\,\Dmineps(\rho_1||\tau^H),
\end{align}
where $\Dmineps$ is the smooth min-relative entropy \cite{Datta09} and $\eps \geq 0$ is the allowed failure probability of the process. Similarly, in Step $c$ the final state $\eta^H$ is formed from the thermal state by applying a protocol that costs work. This work is \cite{HO13}
\begin{align}
	W\upix c = -k_BT\,\ln2\,\Dmaxeps(\eta^H||\tau^H),
\end{align}
where $\Dmaxeps$ is the smooth max-relative entropy \cite{Datta09}. In total, the single-shot work associated to Steps $b$ and $c$ of the process is $W\upix b + W\upix c = k_BT\,\ln2\,(\Dmineps(\rho_1||\tau^H) - \Dmaxeps(\eta^H||\tau^H))$ with failure probability at most $2\eps-\eps^2 \approx 2\eps$, when $\eps$ is small. 

To show consistency we now consider the average expected work extracted per copy if the single-shot protocol is carried out on $N\rightarrow\infty$ i.i.d. copies of the system. In such a calculation the work computed is an average value which is why $\avg{W^{(a)}}$, the average work contribution of the basis rotation in Step $a$, can be taken into account too. One obtains a total average work per copy of
\begin{align}\begin{split} \label{eq:avgopt}
	\avg{W} &= \avg{W^{(a)}} + k_BT\,\ln2\, \lim_{\eps\rightarrow0}\lim_{N\rightarrow\infty} \frac{1}{N} \\
	& \quad \left[\, \Dmineps\big(\rho_1^{\otimes N}||(\tau^H)^{\otimes N}\big) - \Dmaxeps\big((\eta^H)^{\otimes N}||(\tau^H)^{\otimes N}\big) \,\right] \\
	&= \avg{W^{(a)}} + k_BT\,\ln2\,\left[\, D(\rho_1||\tau^H) - D(\eta^H||\tau^H)\,\right] \\
	&= \avg{W^{(a)}} + \tr[\rho_1\, H] +\ln Z -k_BT\,S(\rho_1) \\	
	&\quad - \tr[\eta^H\, H] - \ln Z + k_BT\, S(\eta^H)  \\
		&= k_BT\, (S(\eta^H) - S(\rho))  \\
		&\equiv \avg{W_\opt},
\end{split}\end{align}
where we have used the quantum asymptotic equipartition theorem for relative entropies \cite{Tomamichel12, GenEntropies} in the second line. $D(\cdot||\cdot)$ is the standard quantum relative entropy defined by $D(\eta^H||\tau^H) = \tr[\,\eta^H(\log\eta^H - \log \tau^H)\,]$ and likewise for $\rho_1$, where $\log$ is the logarithm to base 2. The quantities $\Dmineps$ and $\Dmaxeps$ as well as their regularized version, the standard quantum relative entropy $D$, can be seen as different measures characterizing the distance between two states. When applied here, they measure the `distance' between the thermal state $\tau^H$ and another diagonal state in such a way that the operational meaning of this distance is given by the work one has to invest or is able to extract when transforming one into the other.

The derivation shows that in the asymptotic limit the optimal average work is recovered from the single-shot components. But it is important to realise that from Eq.~(\ref{eq:avgopt}) one cannot conclude that the above single-shot process forming $\eta^H$ from $\rho_1$ is optimal. Going via the thermal state is just one option which is particularly convenient in this case as the processes of maximal work extraction and work of formation from the thermal state have been treated in the single-shot scenario \cite{HO13}. It is an open question whether there are better single-shot protocols for general thermodynamic transformations, see Fig.~\ref{fig:singleshotmmt}b\ \&\ \ref{fig:singleshotmmt}c. {The introduction of ``catalysts'' in single-shot thermodynamics \cite{Brandao13b} provides a promising avenue to establish bounds on the work that can be drawn from a state with coherences during a projection in the single-shot setting. }

After making public our results on average work associated with removing coherences in thermodynamic projection processes very recently a paper appeared \cite{David15} that derives the work that can be extracted when removing coherences in a single-shot setting. In this paper the previously mentioned framework describing the catalytic role of coherence in thermodynamics by {\AA}berg \cite{Aberg14} is used together with insights from reference frames in quantum information theory. These results are in agreement with our findings and strengthen our conclusion that coherences are a fundamental feature distinguishing quantum from classical thermodynamics.

\section{Quantum work fluctuation relation} \label{app:jarzynski}

A common route of deriving the quantum Jarzynski equation is as follows \cite{Mukamel03,TLH07,MT11}. A quantum system is initialised in a thermal state $\rho_0 = \sum_n \, e^{-\beta (E\upix 0_n - F\upix 0_T)}  \, \Pi\upix 0_n$ for a given Hamiltonian $H\upix 0 = \sum_{n} E\upix 0_n \, \Pi\upix 0_n$, with energy eigenvector projectors $\Pi\upix 0_n =  \ket{e\upix 0_n} \bra{e\upix 0_n}$, at given inverse temperature $\beta =1/(k_B \, T)$. Here $F\upix 0_T= -k_BT\,\ln \left(\sum_n \, e^{-\beta E\upix 0_n}\right)$ is the initial free energy associated with the initial Hamiltonian $H\upix 0$. The aim is to calculate the average exponentiated work, $\avg{e^{\beta W} }$, that the quantum system will exchange when undergoing a unitary process $V$ that is generated by varying the Hamiltonian in time, i.e. $H\upix t$, from $ H\upix 0$ to a final $H\upix \tau = \sum_{m} E\upix \tau_m \, \Pi\upix \tau_m $. The final state after the unitary is the non-equilibrium state $\rho_{\tau} := V \, \rho_0 \, V^{\dag}$, see Fig.~\ref{fig:QuJarzynski}a. 

To identify the work for an individual run of the experiment the energy of the system is measured at the beginning, by projecting into $\Pi\upix 0_n$, and at the end, by projecting in the final energy basis $\Pi\upix \tau_m = \ket{e\upix \tau_m} \bra{e\upix \tau_m}$. The (extracted) fluctuating work $W$ identified with each transition $\ket{e\upix 0_n} \to \ket{e\upix \tau_m}$ is the (negative) observed \emph{fluctuating energy difference} of the system
\begin{align} \label{eq:single-work}
 	\Delta E =  E\upix \tau_m - E\upix 0_n = - W.
\end{align}
The average exponentiated work then becomes
\begin{align}
    \avg{e^{\beta W} } = \sum_{m,n} e^{- \beta (E\upix \tau_m - E\upix 0_n)}  \, p^{\tau}_{m,n},
\end{align}
where $p^{\tau}_{m,n}$ are the transition probabilities for energy jumps starting in $\ket{e\upix 0_n}$ and ending in $\ket{e\upix \tau_m}$ at time $\tau$. These probabilities are given by
\begin{align}\begin{split}
    p^{\tau}_{m,n} &= \tr[\,\Pi\upix \tau_m \, V \, \Pi\upix 0_n \, \rho_0 \, \Pi\upix 0_n \, V^{\dag} \, \Pi\upix \tau_m\,]\\
    &= e^{-\beta (E\upix 0_n - F\upix 0_T)} \, \tr[\,\Pi\upix \tau_m \, V \, \Pi\upix 0_n \, V^{\dag}\,],
\end{split}\end{align}
simplifying the exponentiated average work to 
\begin{align}
    \avg{e^{\beta W} } = e^{\beta F\upix 0_T} \, \sum_{m} e^{-\beta E\upix \tau_m}  \, \sum_n \, \tr[\,\Pi\upix \tau_m \, V \, \Pi\upix 0_n \, V^{\dag}\,].
\end{align}
The completeness of the projectors, $\sum_n \Pi\upix 0_n = \mathbbm{1}$, now finally results in the well-known quantum Jarzynski work relation
\begin{align}
    \avg{e^{\beta W} } = e^{\beta F\upix 0_T} \, \sum_{m} e^{-\beta E\upix \tau_m} = e^{- \beta \Delta F},
\end{align}
where $\Delta F = F\upix \tau_T - F\upix 0_T$ is the difference of the equilibrium free energies corresponding to the final and initial Hamiltonians, i.e. $F\upix \tau_T = -k_BT\,\ln \left(\sum_m \, e^{-\beta E\upix \tau_m}\right)$. 
Similarly, the average work extracted from the system is the average energy difference between $\rho_0$ and $\rho_{\tau}$
\begin{align}\begin{split}
    \avg{W^{\rm unitary}} &= - \sum_{m,n} {(E\upix \tau_m - E\upix 0_n)}  \, p^{\tau}_{m,n} \\
    &= - \sum_{m} E\upix \tau_m \, p^{\tau}_{m}  + \sum_{n} E\upix 0_n \, p^{0}_{n} \\
	&= - \tr[H\upix \tau \, \rho_\tau] + \tr[H\upix 0_n \, \rho_0] \\
	&= - U(\rho_\tau) + U(\rho_0)
\end{split}\end{align}
where $p^{\tau}_{m} := \sum_n p^{\tau}_{m,n} =  \tr[ \rho_\tau \, \Pi\upix \tau_m\,]$ are the probabilities to find energies $E\upix \tau_m$ when a measurement of $H\upix \tau$ is performed on $\rho_\tau$. By construction, the initial probabilities to find energies $E\upix 0_n$ are just the thermal probabilities, $p^{0}_{n} := \sum_m p^{\tau}_{m,n} = \tr[ e^{-\beta (E\upix 0_n - F\upix 0_T)}  \, \Pi\upix 0_n] = e^{-\beta (E\upix 0_n - F\upix 0_T)}$. 

The derivations of the average work, $\avg{W^{\rm unitary}} = U(\rho_0) - U(\rho_\tau)$, as well as the average exponentiated work, i.e. the quantum Jarzynski equality, $\avg{e^{\beta W} } = e^{- \beta \Delta F}$, are based on Eq.~(\ref{eq:single-work}) which was made assuming that the initial and final state of the process that is being characterised are $\rho_0$ and $\rho_{\tau}$.
There is no question that the mathematical details of the derivations of the above relation are sound. Experimentally, there is however a need to acquire knowledge of the fluctuating energy to quantify the work and this requires the implementation of the second measurement, see Fig.~\ref{fig:QuJarzynski}b. Only after the measurement has been made can theoretical predictions be tested. The measurement is  an unavoidable non-unitary component of the overall experimental process. Specifically, the ensemble state after the unitary, $\rho_{\tau}$, is further altered by the measurement to result in the final state $\eta_{\tau}$, i.e. it is the state $\rho_{\tau}$ with any coherences in the energy basis of $H\upix {\tau}$ removed.  

While the experimentally observed average energy difference is not affected by the measurement step, i.e. $U(\eta_{\tau}) - U(\rho_0) = U(\rho_{\tau}) - U(\rho_0)$, the entropy difference does change, i.e. $S(\eta_{\tau}) - S(\rho_0) \not = S(\rho_{\tau}) - S(\rho_0) =0$. This means that the system may absorb heat, $\avg{Q^{\rm abs}}$,  during the measurement step, indicated in Fig.~\ref{fig:QuJarzynski}b. Its actual value depends on \emph{how} the measurement is conducted with the optimal heat positive, $\avg{Q^{\rm abs}_{\rm opt}} = k_B\,T (S(\eta_{\tau}) - S(\rho_\tau)) \geq 0$. Since $\Delta U = \avg{Q^{\rm abs}} - \avg{W}$ (T1) this implies that in an experimental implementation of the Jarzynski relation the work done by the system on average can be more than previously thought, with the optimal value being $\avg{W_{\rm opt}} = \avg{W^{\rm unitary}} + k_B \, T \, (S(\eta_{\tau}) - S(\rho_{\tau}))$. In the special case that the average heat $\avg{Q^{\rm abs}}$ is zero it is possible (although not necessary) that Eq.~(\ref{eq:single-work}), and thus the standard Jarzynski expression $\avg{e^{\beta W} } = e^{- \beta \Delta F}$, are correct. In particular this applies to classical measurements. We conclude that the suitability of identifying $W = - \Delta E$, and hence the validity of the quantum Jarzynski work relation depends on the details of the physical process that implements the measurement. 

Quantum work fluctuation relations  that have only one measurement \cite{Mazzola13, Roncaglia14}, instead of the two discussed above, offer a feasible route of measuring work fluctuations experimentally. Instead of measuring separately the initial and final fluctuating energies, $E\upix 0_n$ and $E\upix \tau_m$, to establish their joint probabilities, this method acquires \emph{only} knowledge of the joint probabilities by measuring energy differences $\Delta E$ directly. But also here is one final measurement, in general on a non-diagonal state, needed.

\section{Access to correlated auxiliary systems} \label{app:correlations}

Similarly to erasure with a correlated memory \cite{delRio} one can consider projections on a system $S$ that is correlated with an ancilla $A$ the experimenter has access to. Assuming a total Hamiltonian $H^{SA} = H^S\otimes\mathbbm{1}^A + \mathbbm{1}^S\otimes H^A$, we denote the global initial state by $\rho^{SA}$ and its marginals on $S$ and $A$ by $\rho^S = \tr_A[\rho^{SA}]$ and $\rho^A = \tr_S[\rho^{SA}]$, respectively. \\

\textbf{A note on notation.} For clarity we employ a slightly different notation here. The roles of initial state $\rho$ and final state $\eta$ are the same as in the main text and the previous sections of the Appendix. However, now the superscripts of the final state $\eta$ no longer denote the projection basis but the system for which $\eta$ describes the state. For instance, $\eta^{S}$ denotes the reduced state after the projection on system $S$ alone. The same holds for the superscript of the initial state, $\rho^{SA}, \rho^S$ and $\rho^A$, and the Hamiltonians $H^{SA}, H^S$ and $H^A$. Only the superscript ${\cal P}$ of the mutually orthogonal rank-1 projectors $\{ \Pi^{\cal P}_k\}_k$ acting on system $S$ is kept to indicate which basis is being projected in. \\

For an initial global state $\rho^{SA}$ of system and ancilla a local projection map on $S$ results in a new global state
\begin{align}
\eta^{SA} = \sum_k \left(\Pi^{\cal P}_k\otimes\mathbbm{1}^A \right)\, \rho^{SA} \, \left(\Pi^{\cal P}_k\otimes\mathbbm{1}^A\right).
\end{align}
Due to the properties of the projectors the marginal state on $A$ is unchanged,
\begin{align}\begin{split}
\eta^A &= \tr_S[ \eta^{SA}] = \sum_k \tr_S[(\Pi^{\cal P}_k\otimes\mathbbm{1}^A)\, \rho^{SA} \, (\Pi^{\cal P}_k\otimes\mathbbm{1}^A)] \\
&= \sum_k \tr_S[\rho^{SA} \, (\Pi^{\cal P}_k\otimes\mathbbm{1}^A)] = \tr_S[ \rho^{SA}] = \rho^A.
\end{split}\end{align}
The reduced state of the system becomes $\eta^{S} \equiv \tr_A[\eta^{SA}] = \sum_k \Pi^{\cal P}_k \, \rho^{S} \, \Pi^{\cal P}_k = \sum_k p_k \, \Pi^{\cal P}_k$ where $p_k = \tr[\rho^{S} \Pi^{\cal P}_k]$, and the conditional states on $A$ after the process are denoted $\eta^A_k = p_k^{-1}\, \tr_S[\,(\Pi^{\cal P}_k\otimes\mathbbm{1}^A) \,  \rho^{SA} \, (\Pi^{\cal P}_k\otimes \mathbbm{1}^A)\, ]$ for all $k$. The global entropy change associated with the local projection is
\begin{align}\begin{split}
\label{eq:deltaSSA}
	\Delta {\cal S}^{\cal P} &= S(\eta^{SA}) - S(\rho^{SA}) \\
	&= S(\{p_k\}) + \sum_k p_k S(\eta^A_k) - S(\rho^{SA}) \\
	&= S(\eta^S) - S(\rho^S) + S(\rho^{S}) + \sum_k p_k S(\eta^A_k) - S(\rho^{SA}) \\
	&= \Delta S^{\cal P} + \delta^{\cal P}(A:S).
\end{split}\end{align}
In the second equality it was used that $\eta^{SA}$ is a classical-quantum-state and $S(\{p_k\}) = -\sum_k p_k\ln p_k$ stands for the classical Shannon entropy \cite{Shannon48} which is equal to the von Neumann entropy of $\eta^S$ because the final state on $S$ is a classical mixture of states from the projective basis. Here we defined a measure of correlations between the ancilla and the system, $\delta^{\cal P}(A:S) = S(\rho^{S}) - S(\rho^{SA}) + \sum_k p_k S(\eta^A_k)$, related to the quantum discord. It depends on the projectors $\{ \Pi^{\cal P}_k \}_k$ and is always positive \cite{discord, OllivierZurek02}. Thus the entropy change of $SA$ can be bigger than the local entropy change, $\Delta S^{\cal P} =S(\eta^S) - S(\rho^S)$, on the system alone.

As shown before, Eq.~(\ref{eq:Wmeasuregeneral}), the optimal extractable work in a thermodynamic projection process on system $S$ alone is $\avg{W_\opt} = k_B\, T\, \Delta S^{\cal P} - \Delta U^{\cal P}$, where $\Delta S^{\cal P}$ is the entropy change of the system and $\Delta U^{\cal P}$ its change in internal energy. This result stays intact when generalizing to projections in the presence of ancillary systems if one takes the total changes of these quantities on $SA$ instead of the change on $S$ only. In the global process the total internal energy change is equal to the energy change of the system only as the local state of the ancilla is unchanged and the total Hamiltonian is the sum of local Hamiltonians. Thus using side information the overall optimal extractable work amounts to
\begin{align}\begin{split}
\avg{{\cal W}_\opt}  &= k_B T\, \Delta {\cal S}^{\cal P} - \Delta U^{\cal P} \\
&= k_B T\, \Delta S^{\cal P}  + k_BT \, \delta^{\cal P} (A:S) - \Delta U^{\cal P}\\
&= \avg{W_\opt} + k_BT \, \delta^{\cal P} (A:S),
\end{split}\end{align}
where $\avg{W_\opt}$ is the work of an optimal thermodynamic projection process without access to correlated systems, Eq.~(\ref{eq:Wmeasuregeneral}).
 
Discord was first discussed in a thermodynamic context by Zurek \cite{Zurek03}, where he related it to the advantage a quantum Maxwell demon could have over a classical one. In general the quantum discord, $\delta(A:S)$, is defined as the minimum of $\delta^{\cal P}(A:S)$ over all sets of projectors $\{ \Pi^{\cal P}_k \}_k$ whereas in our case this set is fixed (see e.g. Modi \textit{et al.} \cite{ModiReview} for a review). Therefore it is found that even for states with no quantum discord, usually referred to as classically correlated states, a difference in work associated with thermodynamic projection processes can be observed. This contrasts with the erasure process \cite{delRio} where an advantage could only be gained for highly entangled states.

One may ask what global states on $SA$ maximize $\avg{{\cal W}_\opt}$ for a given state $\rho^S$ on $S$. Expectedly, it can be shown that purifications of $\rho^S$ yield the best improvement in terms of extracted work. Given $\rho^S = \sum_l a_l \pure{l}$ any purification is, up to isometries on the purifying system \cite{NielsenChuang}, equivalent to $\ket{\Psi}=\sum_l \sqrt{a_l} \ket{l}^S\ket{l}^A$ for some orthonormal basis $\{\ket{l}^A\}_l$ of $A$. For such a state the conditional states on $A$ after the projection, $\rho_k^A$, are pure for all $k$ which implies that they have zero entropy. This implies
\begin{align}\begin{split}
\delta^{\cal P}(A:S) &= S(\rho^S) - S(\pure{\Psi}) +\sum_k p_k S(\eta_k^A) \\
&= S(\rho^S) - 0 + \sum_k p_k\cdot0 = S(\rho^S). 
\end{split}\end{align}
The optimal total extracted work from a purified state on $SA$ in a thermodynamic projection  process is therefore $\avg{{\cal W}_\opt} = k_BT\, S(\eta^S) - \Delta U^{\cal P}$ which can be shown to be the maximum for fixed $\rho^S$ and projectors $\{ \Pi_k^{\cal P} \}_k$. One way to see this is the following Lemma \cite{Renato}. 

\begin{lemma}
Let $\rho^{SA}$ be an arbitrary state on a bipartite system $SA$, and let $\{\Pi^{\cal P}_k\}_k$ be a complete set of rank-1 orthogonal projectors on $S$. Consider the global state after the projection, $\eta^{SA} = \sum_k (\Pi^{\cal P}_k\otimes \mathbbm{1}^A) \,  \rho^{SA} \, (\Pi^{\cal P}_k\otimes \mathbbm{1}^A) = \sum_k p_k\, \Pi^{\cal P}_k\otimes\eta^A_k$, where $p_k = \tr[(\Pi^{\cal P}_k\otimes \mathbbm{1}^A) \, \rho^{SA}]$ are the probabilities to measure $k$ on $S$ and $\eta^A_k = p_k^{-1}\, \tr_S[ (\Pi^{\cal P}_k\otimes \mathbbm{1}^A)\,  \rho^{SA} \, (\Pi^{\cal P}_k\otimes \mathbbm{1}^A) ]$ are the conditional states on $A$. Then
\begin{align}
S(\rho^{SA}) \, \geq \, \sum_k p_k S(\eta_k^A).
\end{align}
\end{lemma}

\begin{proof}
We model the process on $S$ as an isometry $\Phi^{S\rightarrow S\tilde{S}} = \sum_k \ket{\psi_k}^{\tilde{S}}\otimes \Pi_k^S$, where $\tilde{S}$ is a copy of $S$ and $\{\ket{\psi_k}^{\tilde{S}}\}_k$ is an orthonormal basis of $\tilde{S}$. The state after applying the isometry is denoted $\eta^{S\tilde{S}A} = \Phi\rho^{SA}\Phi^\dagger$ and we note that $\tr_{\tilde{S}}[\eta^{S\tilde{S}A}] = \eta^{SA}$. Furthermore, isometries do not change (von Neumann) entropy and thus, $S(\eta^{S\tilde{S}A}) = S(\rho^{SA})$. In addition, by construction of $\tilde{S}$ the marginals of the final state on $S$ and $\tilde{S}$ have the same entropy: $S(\eta^{\tilde{S}}) = S(\eta^S)$. Since $S(\eta^{S\tilde{S}A}) \geq |S(\eta^{SA})- S(\eta^{\tilde{S}})| \geq S(\eta^{SA})- S(\eta^{\tilde{S}})$ (see e.g. Nielsen \& Chuang \cite{NielsenChuang}). Thus
\begin{align}\begin{split}
S(\rho^{SA}) &= S(\eta^{S\tilde{S}A}) \geq S(\eta^{SA}) - S(\eta^{\tilde{S}}) \\
&=S(\eta^{SA}) - S(\eta^S) = \sum_k p_k S(\eta_k^A),
\end{split}\end{align}
where in the last equality we made use of the fact that $\eta^{SA}$ is a classical-quantum state.
\end{proof}

Going back to Eq.~(\ref{eq:deltaSSA}) and applying the the above Lemma we see that in general $\avg{{\cal W}_\opt} = k_B T\, \Delta {\cal S}^{\cal P} - \Delta U^{\cal P} = k_BT\, \left[\,S(\eta^S) +\sum_k p_k S(\eta_k^A) - S(\rho^{SA})\,\right] - \Delta U^{\cal P} \leq k_BT \, S(\eta^S)- \Delta U^{\cal P}$, which proves that purifications on $SA$ yield the maximally possible extracted work.


\end{document}